\documentclass[12pt]{article}

\usepackage{header}

\title{\Large{ {\bf Single-Peaked Domain Augmented with Complete Indifference: A Characterization of Target Rules with a Default }}}

\author{Parikshit De\thanks{Economic Sciences, IISER-Bhopal, Bhopal Bypass Road, Bhauri, Bhopal, Madhya Pradesh, India - 462066. Email: \texttt{parikshitde@iiserb.ac.in}},\;Abinash Panda\thanks{Department of Economics, Shiv Nadar Institution of Eminence, NH - 91, Gautam Buddha Nagar, Uttar Pradesh, India - 201314. Email: \texttt{ap280@snu.edu.in}}\; and Anup Pramanik\thanks{Department of Economics, Shiv Nadar Institution of Eminence, NH - 91, Gautam Buddha Nagar, Uttar Pradesh, India - 201314. Email: \texttt{anup.pramanik@snu.edu.in}.}}

\begin{document}

\maketitle
 
\begin{abstract} We study a public decision problem in which a finite society selects a public-good level from a closed interval. Agents either have single-peaked preferences or are completely indifferent over the interval; the latter capture abstention or a “none of the above” stance within the decision process. We study this augmented single-peaked domain. On this domain, we characterize the class of rules called \textit{target rules with a default}. We show that onto-ness and pairwise strategy-proofness characterize this class of rules.
\end{abstract}



\newpage
\section{Introduction}
We study a public decision problem in which a finite society selects a public good level from a closed interval based on agents’ preferences. Agents may either have single-peaked preferences or be completely indifferent over all feasible levels. Completely indifferent agents can be interpreted as abstaining, or as expressing a ``none of the above" (NOTA) position - a refusal to endorse any specific public good level while remaining part of the decision procedure.

We refer to this domain as the single-peaked domain augmented with complete indifference. A social choice function (or a rule) assigns a public-good level to each preference profile. Our main result characterizes a class of rules, which we term as \textit{target rules with a default}. Each such rule specifies a target level and a default level in the interval. At profiles with at least one single-peaked agent, the rule selects the target whenever it is efficient; otherwise, it selects the peak of a single-peaked agent that is closest to the target. When all agents are completely indifferent, the rule selects the default level.

The characterization relies on onto-ness and pairwise strategy-proofness. Onto-ness requires that every feasible public good level be chosen at some preference profile. Pairwise strategy-proofness requires that no coalition of at most two agents can jointly misreport their preferences in a way that makes all coalition members weakly better off, with at least one of them strictly better off. This property is weaker than group strategy-proofness, which excludes profitable manipulations by coalitions of arbitrary size. Pairwise strategy-proofness offers a tractable and behaviorally plausible weakening of group strategy-proofness. Tractability arises because the requirement restricts attention to coalitions of size two, thereby avoiding the combinatorial complexity associated with deviations by arbitrary coalitions. The property is also behaviorally appealing, as coordination among two agents is substantially easier than large-scale coalition formation, which may be informationally or institutionally demanding. Despite this relaxation, the condition remains strong enough to yield a sharp characterization of target rules with a default on the augmented domain.

Our work complements the literature on strategy-proof social choice functions over extensions of the classical single-peaked domain. The study of strategy-proofness on the single-peaked domain was initiated by the seminal contribution of \cite{moulin1980strategy}. Subsequent research established the tops-only property for this domain under various formulations and assumptions; see, for example, \cite{barbera1993generalized}, \cite{barbera1994characterization}, \cite{ching1997strategy}, \cite{weymark2008strategy}, \cite{chatterji2011tops}, and \cite{weymark2011unified}.\footnote{A social choice function satisfies the tops-only property if the chosen alternative depends only on agents' most-preferred alternatives. In the literature it is common to say that a domain exhibits the tops-only property when every strategy-proof social choice function on that domain is tops-only.} A number of papers then considered environments in which the classical domain is enlarged to accommodate indifference. For instance, \cite{berga1998} analyze single-plateaued preferences, where agents may have a flat top rather than a unique peak, while \cite{cantala2004choosing} study a model with an outside option, leading to preferences that are single-peaked on an interval of acceptable alternatives and flat outside that interval. In both the settings, strategy-proofness together with onto-ness does not imply tops-onlyness - a feature that also arises in our framework. Assuming tops-onlyness, \cite{berga1998} characterize strategy-proof rules on the single-plateaued domain, showing that they correspond to \cite{moulin1980strategy}'s min–max rules combined with an appropriate tie-breaking selection within plateaus. By contrast, \cite{cantala2004choosing} assume efficiency and obtain a characterization of strategy-proof rules that are tops-only and form a subclass of Moulin's min–max rules. In our setting, efficiency together with strategy-proofness likewise implies tops-onlyness. Our approach, however, differs in that we retain onto-ness and strengthen strategy-proofness to pairwise strategy-proofness. Under these requirements, we obtain a complete characterization of target rules with a default, which are tops-only, efficient, and anonymous.

Our result is closely related to the literature on solidarity principles in public decision problems with single-peaked preferences. Solidarity requirements state that when the environment changes, agents who are not responsible for the change should be affected in the same direction.\footnote{Solidarity principles originate in the fair division literature. Population monotonicity \cite{thomson1983fair, thomson1983problems} and replacement domination \cite{moulin1987pure} are the two most prominent formulations of this idea.} In public decision problems with single-peaked preferences, such axioms, together with efficiency, lead to target rules - rules indexed by a target alternative that, at each preference profile, select the target whenever it is efficient and otherwise choose the peak closest to the target. For fixed-population environments like ours, \cite{thomson1993replacement} shows that replacement domination together with efficiency characterizes the class of target rules (see also \cite{vohra1999replacement}, \cite{klaus2020solidarity} for extensions). In variable-population settings, \cite{ching1993population} shows that efficiency and population monotonicity characterize the same class. In contrast to this literature, we enlarge the standard single-peaked domain by permitting complete indifference and adopt an incentive-based approach. On this augmented domain, pairwise strategy-proofness and onto-ness characterize the class of target rules with a default, thereby providing an incentive-based foundation that does not rely on solidarity axioms.

It is important to highlight \cite{gordon2007public}, who establishes a deep connection between solidarity principles and incentive properties. In particular, \cite{gordon2007public} studies a public decision problem in a variable-population framework and focuses on efficient social choice functions. He shows that, under replication indifference - a property specific to variable-population models - solidarity axioms are equivalent to group strategy-proofness, and this equivalence holds over a very general preference domain. Even if one combines the results of \cite{gordon2007public} with those of \cite{ching1993population} or \cite{thomson1993replacement}, the resulting implications still do not yield our characterization, for several reasons. First, we work in a fixed-population model, whereas \cite{gordon2007public} operates in a variable-population framework; replication indifference, which plays a central role in his analysis, is inherently tied to variable-population environments and has no direct analogue in our setting. Second, \cite{ching1993population} and \cite{thomson1993replacement} consider the single-peaked domain, whereas we study single-peaked domain augmented with complete indifference. Third, our requirements of onto-ness and pairwise strategy-proofness are weaker than efficiency and group strategy-proofness, respectively. Thus, our characterization provides an independent incentive-based foundation for target rules that does not rely on solidarity axioms.

The remainder of the paper is organized as follows. Section \ref{Model} introduces the model and the basic properties of the rules. Section \ref{MR} presents the main results. Section \ref{WIC} discusses the consequences of weakening pairwise strategy-proofness. Section \ref{Conclusion} concludes. All proofs are relegated to the Appendix.

\section{Model}\label{Model}

Let $A = [0,1]$ be the closed interval of the real line. The elements of $A$ represent feasible levels of a public good and constitute the set of alternatives. Consider a society $N = \{1,2,\ldots,n\}$ with $n\geq 2$ agents who must choose a level of the public good from $A$.

For each agent $i \in N$, let $R_i$ denote agent $i$’s preference over $A$. We assume that $R_i$ is a complete and transitive binary relation on $A$. The strict and indifference components of $R_i$ are denoted by $P_i$ and $I_i$, respectively. Let $\mathcal{R}$ denote the set of all such preferences on $A$. We use $<$ to denote the natural ordering on the real line. A preference $R_i$ is \emph{single-peaked with respect to $<$} if there exists an alternative $\tau(R_i) \in A$, called the \emph{peak} of $R_i$, such that for all $a,b \in A$, if  $\tau(R_i) \leq a < b$  or  $b < a \leq \tau(R_i)$, then  $a\;P_i\; b$. Thus, alternatives closer to the peak are strictly preferred on either side. Let $\mathcal{S} \subset \mathcal{R}$ denote the set of all preferences that are single-peaked with respect to $<$.\footnote{Throughout, $\subset$ denotes strict inclusion, while $\subseteq$ denotes weak inclusion.} Let $R^{0}$ denote the preference on $A$ such that
$a R^{0} b$ for all $a,b \in A$; thus $R^{0}$ represents complete indifference.

The preference domain studied in this paper is $\mathcal{S}_0 = \mathcal{S} \cup \{R^{0}\}$, the common set of admissible preferences for each agent, which we call the \emph{single-peaked domain augmented with complete indifference}. A \emph{preference profile} is an $n$-tuple $R=(R_1,\ldots,R_n)\in \mathcal S_0^{\,n}$. When we want to emphasize agent $i\in N$ or a coalition $S\subseteq N$, we write the profile as 
$(R_i,R_{-i})$ or $(R_S,R_{-S})$, respectively, where $R_S=(R_j)_{j\in S}$ and $R_{-S}=(R_j)_{j\in N\setminus S}$.\footnote{Unless stated otherwise, coalitions are nonempty. We write $i$ for the singleton coalition $\{i\}$ (abusing notation).}
The profile in which every agent has the completely indifferent preference $R^{0}$ is called the \emph{completely indifferent preference profile} and is denoted by $R^{0N}=(R^{0},\ldots,R^{0})$. For any coalition $S\subseteq N$, let $R^{0S}$ denote the profile for agents in $S$ in which every 
agent in $S$ has preference $R^{0}$.

For any preference profile $R \in \mathcal{S}_0^n \setminus \{R^{0N}\}$, define $\tau(R)= \{\, \tau(R_i) \in A \mid R_i \neq R^{0} \,\}$. That is, $\tau(R)$ is the set of peak alternatives of all agents with single-peaked preferences at the profile $R$. Since $R \in \mathcal{S}_0^n \setminus \{R^{0N}\}$, the set $\tau(R)$ is nonempty. For any preference profile $R \in \mathcal{S}_0^n \setminus \{R^{0N}\}$, let $\underline{\tau}(R)= \min \tau(R)$ and $\overline{\tau}(R)= \max \tau(R)$. That is, $\underline{\tau}(R)$ and $\overline{\tau}(R)$ denote, respectively, the smallest and the largest peak among agents with single-peaked preferences at the profile $R$.

We now define social choice functions and the axioms considered in this paper.

\begin{defn}\rm
A social choice function (SCF), or simply a rule, is a mapping
$f : \mathcal{S}_0^n \to A$, which assigns to each preference profile an alternative in $A$.
\end{defn}

We first define incentive compatibility at the individual level.

\begin{defn}\rm 
A social choice function $f$ is \emph{strategy-proof} if for every agent $i \in N$, every preference profile $R \in \mathcal{S}_0^n$, and every alternative preference $R'_i \in \mathcal{S}_0$, it holds that
\[
f(R_i,R_{-i}) \, R_i \, f(R'_i,R_{-i}).
\]
That is, no agent can benefit from misreporting her preference, regardless of the reports of others.
\end{defn}

We next strengthen the incentive requirement to deviations by groups of agents.

\begin{defn}\rm 
A social choice function $f$ is \emph{group strategy-proof} (GSP) if for every coalition $S \subseteq N$ and every preference profile $R \in \mathcal{S}_0^n$, there exists no $R'_S\in \mathcal{S}_0^{|S|}$ such that $f(R'_S,R_{-S}) \, R_i \, f(R_S,R_{-S})$ for all $i \in S$, and $f(R'_S,R_{-S}) \, P_j \, f(R_S,R_{-S})$ for some $j \in S$.
\end{defn}

That is, no coalition can jointly misreport in a way that makes all its members weakly better off and at least one member strictly better off. A weaker requirement restricts attention to deviations by coalitions of at most two agents.

\begin{defn}\rm 
A social choice function $f$ is \emph{pairwise strategy-proof} if the group strategy-proofness condition holds for all coalitions $S \subseteq N$ with $|S| \le 2$.
\end{defn}

We now define efficiency.

\begin{defn}\rm
An alternative $a \in A$ is \emph{efficient} at a preference profile $R \in \mathcal{S}_0^n$ if there exists no alternative $b \in A$, $b \neq a$, such that $b\;R_i\;a$ for all $i \in N$, and
$b\;P_j\;a$ for some $j \in N$. A social choice function $f$ is \emph{efficient} if $f(R)$ is efficient at every $R \in \mathcal{S}_0^n$.
\end{defn}

On $\mathcal{S}_0$, efficiency has a simple structure. For any $R \neq R^{0N}$, the set of efficient alternatives is $[\underline{\tau}(R), \overline{\tau}(R)]$. That is, efficiency requires the chosen outcome to lie between the smallest and the largest peak among agents with single-peaked preferences. At the null preference profile $R^{0N}$, every alternative in $A$ is efficient.

We next introduce a standard richness condition on the range of a rule.

\begin{defn}\rm
A social choice function $f$ is \emph{onto} if for every $a \in A$, there exists a preference profile $R \in \mathcal{S}_0^n$ such that $f(R) = a$.
\end{defn}

Onto-ness is a much weaker requirement than efficiency: while efficiency restricts the outcome selected at every profile, onto-ness merely requires that every alternative be attainable at some profile.

Finally, we introduce the \textit{tops-only} property.

\begin{defn}\rm
A social choice function $f$ satisfies the \emph{tops-only} property if for any two preference profiles
$R, R' \in \mathcal{S}_0^n$ such that
\begin{enumerate}
    \item $\{\, i \in N \mid R_i = R^{0} \,\} = \{\, i \in N \mid R'_i = R^{0} \,\}$, and
    \item $\tau(R_i) = \tau(R'_i)$ for every $i \in N$ with $R_i \neq R^{0}$,
\end{enumerate}
it holds that $f(R) = f(R')$.
\end{defn}

In words, the tops-only property requires that the outcome depend only on the set of agents with completely indifferent preferences and on the peak alternatives reported by the remaining agents, and not on any other features of their preferences.

\section{Main Results}\label{MR}
This section presents the main characterization result of the paper. In particular, we characterize social choice functions on the single-peaked domain augmented with complete indifference that satisfy onto-ness and pairwise strategy-proofness. We show that these properties jointly characterize a class of rules that we refer to as \emph{target rules with a default}.

A target rule with a default specifies a target level and a default level in $A$. At preference profiles with at least one single-peaked agent, the rule selects the target whenever it is efficient and otherwise selects the peak closest to the target. When all agents have the completely indifferent preference, the rule selects the default level.

We now provide a formal definition of target rules with a default.

\begin{defn}\rm
A social choice function $f : \mathcal{S}_0^n \to A$ is called a \emph{target rule with a default} if there exist a target level $x \in A$ and a default level $y \in A$ such that, for every preference profile $R \in \mathcal{S}_0^n$,
\[
f(R) =
\begin{cases}
x, & \text{if } R \neq R^{0N} \text{ and } x \in [\underline{\tau}(R), \overline{\tau}(R)], \\[0.3em]
\underline{\tau}(R), & \text{if } R \neq R^{0N} \text{ and } x < \underline{\tau}(R), \\[0.3em]
\overline{\tau}(R), & \text{if } R \neq R^{0N} \text{ and } \overline{\tau}(R) < x, \\[0.3em]
y, & \text{if } R = R^{0N}.
\end{cases}
\]
\end{defn}

We are now ready to state the main theorem.

\begin{theorem}\label{thm}\rm
Let $n\geq 3$. A social choice function $f : \mathcal{S}_0^n \to A$ is onto and pairwise strategy-proof if and only if it is a target rule with a default.
\end{theorem}

The proof of Theorem~\ref{thm} appears in the Appendix. We conclude this section with several remarks.

\begin{remark}\rm
Target rules with a default are efficient and satisfy the tops-only property. Importantly, Theorem~\ref{thm} does not assume efficiency or the tops-only property; both follow from ontoness and pairwise strategy-proofness.
Moreover, target rules are anonymous in the sense that the outcome depends only on agents’ preferences and not on their identities.\footnote{Formally, a social choice function is \emph{anonymous} if permuting the identities of agents does not change the selected outcome.} 
\end{remark}

\begin{remark}\rm
Target rules with a default are group strategy-proof, as shown in Lemma~\ref{L2}. Consequently, in Theorem~\ref{thm}, pairwise strategy-proofness can be replaced by group strategy-proofness without affecting the conclusion. 
\end{remark}

\begin{remark}\rm
Theorem~\ref{thm} assumes that $n\ge 3$. When $n=2$, there exist social choice functions that are onto and pairwise strategy-proof but are not target rules with a default. We illustrate this observation with the following example.

Consider the social choice function $f^d:\mathcal{S}_0^2\to A$ defined by
\[
f^d(R)=
\begin{cases}
\tau(R_1), & \text{if } R_1\neq R^{0}, \\[0.3em]
\tau(R_2), & \text{if } R_1=R^{0}\text{ and }R_2\neq R^{0}, \\[0.3em]
0, & \text{if } R=(R^{0},R^{0}).
\end{cases}
\]
It is immediate that $f^d$ is onto and strategy-proof. Hence, in order to establish pairwise strategy-proofness, it suffices to show that the two agents cannot jointly misreport their preferences in such a way that both become weakly better off and at least one becomes strictly better off.

Suppose, for a contradiction, that such a joint deviation is possible. Then there exist profiles $R=(R_1,R_2)$ and $R'=(R'_1,R'_2)$ such that, according to the preference profile $R$, each agent weakly prefers the outcome $f^d(R')$ to the outcome $f^d(R)$, and at least one agent strictly prefers $f^d(R')$ to $f^d(R)$.

Observe that $R\neq (R^0,R^0)$, since at the profile $(R^0,R^0)$ both agents are indifferent among all alternatives and hence no deviation can make any agent strictly better off. Also, it follows that $f^d(R')\neq f^d(R)$.

First suppose that $R_1\neq R^{0}$. Then, by definition,
$f^d(R)=\tau(R_1)$, which is agent $1$'s most preferred alternative. Since $f^d(R')\neq f^d(R)$, it follows that
$f^d(R)\,P_1\,f^d(R')$, a contradiction.

Next suppose that $R_1=R^{0}$. Then necessarily $R_2\neq R^{0}$ and $f^d(R)=\tau(R_2)$. Since $\tau(R_2)$ is agent $2$'s most preferred alternative and $f^d(R')\neq f^d(R)$, we obtain
$f^d(R)\,P_2\,f^d(R')$, a contradiction. Thus, a profitable joint deviation by the two agents is not possible, and hence $f^d$ is pairwise strategy-proof.

Finally, we show that $f^d$ is not a target rule with a default. Consider profiles $(\bar R_1,\bar R_2)$ and $(R_1^*,R_2^*)$ such that
$\bar R_1,\bar R_2\neq R^{0}$ and $R_1^*,R_2^*\neq R^{0}$, and
\[
\tau(\bar R_1)=\tau(R_2^*)=0
\quad \text{and} \quad
\tau(\bar R_2)=\tau(R_1^*)=1.
\]
Under any target rule with a default, the outcomes at these two profiles must coincide. However,
\[
f^d(\bar R_1,\bar R_2)=0 \neq 1=f^d(R_1^*,R_2^*).
\]
Therefore, $f^d$ is not a target rule with a default.
\end{remark}

\begin{remark}\rm
If the requirement of pairwise strategy-proofness is weakened further - for instance, to strategy-proofness or to weaker variants of group strategy-proofness, then the conclusions of Theorem~\ref{thm} fail to hold. In particular, there exist onto rules that fail to be efficient, tops-only, or anonymous. We analyze these failures in detail in the next section.
\end{remark}

\section{Failure of Efficiency and Top-Onlyness under Weaker Incentive Constraints}\label{WIC}

In this section, we examine how the conclusions of Theorem~\ref{thm} change when the pairwise strategy-proofness requirement is weakened. In particular, we consider individual strategy-proofness and weaker forms of group strategy-proofness, and show that these incentive constraints no longer ensure efficiency or the tops-only property.

We first introduce a weak form of group strategy-proofness that rules out only deviations making every member of a coalition strictly better off.

\begin{defn}\rm 
A social choice function $f$ is \emph{weakly group strategy-proof} (WGSP) if for every coalition $S \subseteq N$ and every preference profile $R \in \mathcal{S}_0^n$, there exists no $R'_S\in \mathcal{S}_0^{|S|}$ such that
\[
f(R'_S,R_{-S}) \, P_i \, f(R_S,R_{-S}) \quad \text{for all } i \in S.
\]
\end{defn}

A further weakening restricts attention to coalitions of size at most two.

\begin{defn}\rm 
A social choice function $f$ is \emph{weakly pairwise strategy-proof} if the weak group strategy-proofness condition holds for all coalitions $S \subseteq N$ with $|S| \le 2$.
\end{defn}

The following example shows that an onto and weakly group strategy-proof rule need not be efficient or tops-only.

\begin{example}\label{ex1}\rm
Consider the social choice function $f^*:\mathcal{S}_0^n \to A$ defined as follows. For every $R \in \mathcal{S}_0^n$,
\[
f^*(R)=
\begin{cases}
\tau(R_1), & \text{if } R_1 \neq R^{0}, \\[0.3em]
0, & \text{if } R_1 = R^{0} \text{ and } 0 \, P_2 \, 1, \\[0.3em]
1, & \text{otherwise.}
\end{cases}
\]
It is straightforward to verify that $f^*$ is onto. We now show that $f^*$ is weakly group strategy-proof. Suppose, to the contrary, that $f^*$ is not weakly group strategy-proof. Then there exist a coalition $S \subseteq N$, a profile $R \in \mathcal{S}_0^n$, and $R'_S \in \mathcal{S}_0^{|S|}$ such that
\[
f^*(R'_S,R_{-S}) \, P_i \, f^*(R) \quad \text{for all } i \in S.
\]
This implies that, $f^*(R'_S,R_{-S}) \neq f^*(R)$ and $R_i \neq R^0$ for all $i \in S$. Observe first that $S$ cannot be a subset of $N \setminus \{1,2\}$, since in that case the definition of $f^*$ implies $f^*(R'_S,R_{-S}) = f^*(R)$, a contradiction. Hence, either $1 \in S$, or $1 \notin S$ and $2 \in S$.

If $1 \in S$, then by definition of $f^*$, $f^*(R) \, P_1 \, f^*(R'_S,R_{-S})$, which contradicts the assumption.

Next, consider the case in which $1\notin S$ and $2\in S$. Then either
\[
f^*(R)=0 \ \text{and}\ f^*(R'_S,R_{-S})=1,
\]
or
\[
f^*(R)=1 \ \text{and}\ f^*(R'_S,R_{-S})=0.
\]
In either case, we have $f^*(R)\,R_2\,f^*(R'_S,R_{-S})$, which yields a contradiction. Therefore, $f$ is weakly group strategy-proof.

However, $f^*$ fails to be efficient and does not satisfy the tops-only property. 
To see the failure of efficiency, consider the profile $R \in \mathcal{S}_0^n$, where $R_i=R^{0}$ for all $i\neq 2$, $R_2\neq R^{0}$, $\tau(R_{2})=0.5$, and $1 \, P_{2} \, 0$. 
At $R$, agent~2 strictly prefers $0.5$ to $1$, while every agent $i\in N\setminus\{2\}$ is indifferent between $0.5$ and $1$. Hence, $1$ is not an efficient outcome at $R$, yet $f^*(R)=1$. 
To see that $f^*$ violates the tops-only property, consider another profile $R'\in \mathcal{S}_0^n$ such that $R'_i=R^{0}$ for all $i\neq 2$, $R'_2\neq R^{0}$, $\tau(R'_{2})=0.5$, and $0 \, P'_2 \, 1$. Note that
\[
\{\, i \in N \mid R_i = R^{0} \,\} = \{\, i \in N \mid R'_i = R^{0} \,\}=N\setminus\{2\}
\quad \text{and} \quad
\tau(R_2)=\tau(R'_2),
\]
yet $f^*(R)=1 \neq 0 = f^*(R')$. Therefore, $f^*$ does not satisfy the tops-only property.
\hfill$\blacksquare$
\end{example}

Efficiency substantially strengthens the implications of incentive compatibility. The following proposition shows that any efficient and strategy-proof rule is necessarily tops-only and weakly group strategy-proof.

\begin{prop}\rm \label{prop}
Let $f:\mathcal{S}_0^n \rightarrow A$ be an efficient and strategy-proof social choice function. Then:
\begin{enumerate}
\item[(a)] $f$ is tops-only.
\item[(b)] $f$ is weakly group strategy-proof.
\end{enumerate}
\end{prop}

The proof of Proposition~\ref{prop} is provided in the Appendix. We conclude this section with the following remarks on Proposition~\ref{prop} and Example~\ref{ex1}.

\begin{remark}\rm
Example~\ref{ex1} illustrates the difficulty of characterizing onto and strategy-proof social choice functions on the domain $\mathcal{S}_0$. In particular, there exist onto and weakly group strategy-proof rules that are neither efficient nor tops-only. This stands in sharp contrast to the single-peaked domain $\mathcal{S}$, where for strategy-proof rules, ontoness and efficiency are equivalent. Moreover, on $\mathcal{S}$, onto and strategy-proof rules are tops-only and even group strategy-proof. The analysis of strategy-proof rules on the domain of single-peaked preferences was initiated by the seminal contribution of \cite{moulin1980strategy}. Subsequent work established the tops-only property for this domain under various formulations and assumptions; see, for example, \cite{barbera1993generalized}, \cite{barbera1994characterization}, \cite{ching1997strategy}, \cite{weymark2008strategy}, \cite{chatterji2011tops} and \cite{weymark2011unified}.
\end{remark}

\begin{remark}\rm
Although weak group strategy-proofness is formally stronger than individual strategy-proofness, Proposition~\ref{prop} implies that, on our domain, the class of efficient and strategy-proof social choice functions coincides with the class of efficient and weakly group strategy-proof social choice functions. Moreover, all such rules satisfy the tops-only property. We do not pursue a characterization of this class in the present paper, as doing so would require introducing additional notation and structure that would add little insight relative to our main objective. It is nevertheless worth noting that \citet{berga1998} provides a characterization of tops-only and strategy-proof rules on the single-plateaued domain. By focusing on those rules in Berga’s characterization that are efficient, then restricting these rules to our domain - namely, to preference profiles in which each agent’s plateau is either a singleton or the entire interval $[0,1]$ - yields a class of social choice functions that are efficient and strategy-proof on our domain. This class may in fact coincide with the full class of efficient and strategy-proof rules on our domain, but we do not investigate this issue further in the present paper. Our primary focus instead is on the characterization of target rules with a default.
\end{remark}

\section{Conclusion}\label{Conclusion}
We study a public decision problem in which agents' preferences are either single-peaked or completely indifferent over a closed interval of public-good levels. On this augmented domain, we characterize the class of onto and pairwise strategy-proof social choice functions, showing that they coincide with the target rules with a default. Thus, we provide an incentive-based foundation for this class that contrasts with solidarity-based characterizations on the classical single-peaked domain.

\bibliographystyle{ecta}
\bibliography{pairwise}
\appendix

\section{Appendix} \label{app1}
\subsection{Proof of Proposition \ref{prop}}
We begin by recalling a established result on strategy-proof social choice functions on the single-peaked domain $\mathcal{S}$. The implication that efficiency and strategy-proofness yield the tops-only property is known under several alternative formulations of single-peaked preferences. For our purposes, we refer to \cite{weymark2011unified}.

Let $h:\mathcal{S}^n \to A$ be a social choice function defined on $\mathcal{S}$. A social choice function $h$ is \emph{strategy-proof} if for every agent $i \in N$, every preference profile $R \in \mathcal{S}^n$, and every $R'_i \in \mathcal{S}$,
$h(R_i,R_{-i}) \; R_i \; h(R'_i,R_{-i})$.

An alternative $a \in A$ is \emph{efficient} at a preference profile $R \in \mathcal{S}^n$, if there exists no alternative $b \in A$, $b \neq a$, such that
\[
b R_i a \quad \text{for all } i \in N,
\]
and
\[
b P_j a \quad \text{for some } j \in N.
\]
A social choice function $h$ is \emph{efficient} if $h(R)$ is efficient at every $R \in \mathcal{S}^n$.

If we consider the following weaker notion of efficiency, it coincides with efficiency on the domain $\mathcal{S}$. An alternative $a \in A$ is \emph{efficient$^{*}$} at a preference profile $R \in \mathcal{S}^n$ if there exists no alternative $b \in A$, $b \neq a$, such that $b P_i a$ for all $i \in N$. However, this equivalence need not hold on our domain $\mathcal{S}_0$.

 Finally, $h$ satisfies the \emph{tops-only property} if for any two preference profiles $R, R' \in \mathcal{S}^n$ such that $\tau(R_i) = \tau(R'_i)$ for every $i \in N$, it holds that $h(R) = h(R')$.

We are now ready to state the result that we will use in the proof of Proposition~\ref{prop}.

\begin{lemma}\rm \label{L1}
Let $h:\mathcal{S}^n \rightarrow A$ be an efficient and strategy-proof social choice function. Then $h$ satisfies the tops-only property.
\end{lemma}

\begin{proof}
See Proposition~1 in \cite{weymark2011unified} for a detailed proof.
\end{proof}

We now proceed to the proof of Proposition~\ref{prop}.
\bigskip

\begin{proof}[Proof of Proposition~\ref{prop}] Let $f:\mathcal{S}_0^n\rightarrow A$ be an efficient and strategy-proof scf.

\begin{enumerate}
\item[(a)]  We will show that $f$ is tops-only. Consider any two distinct profiles $R',R''\in \mathcal{S}_0^n$ such that $\{\, i \in N \mid R'_i = R^{0} \,\} = \{\, i \in N \mid R''_i = R^{0} \,\}$, and $\tau(R'_i) = \tau(R''_i)$ for every $i \in N$ with $R'_i \neq R^{0}$. We will show that $f(R')=f(R'')$.
 
 Let $\{\, i \in N \mid R_i = R^{0} \,\} = \{\, i \in N \mid R'_i = R^{0} \,\}= S$. Note that $S\subset N$, and $S$ may be empty. We consider an social choice function for agents in $N\setminus S$ on single-peaked domain $\mathcal{S}$ as follows. For any $R\in \mathcal{S}^{|N\setminus S|}$,
 $$h^S(R)=f(R^{0S},R).$$
 Since $f$ is strategy-proof and efficient, it follows that $h^S$ is also strategy-proof and efficient. Hence, by Lemma \ref{L1}, it follows that $h^S$ is tops-only. Now consider two distinct profiles $\bar{R}, \hat{R}\in \mathcal{S}^{|N\setminus S|}$ such that $\bar{R_i}=R'_i$ and  $\hat{R_i}=R''_i$ for all $i\in N\setminus S$. Note that $\tau(\bar{R_i})=\tau(\hat{R_i})$ for all $i\in N\setminus S$. Hence by tops-onlyness property of $h^S$, we have that $h^S(\bar{R})=h^S(\hat{R})$. Therefore, $f(R')=f(R^{0S},\bar{R})=f(R^{0S},\hat{R})=f(R'')$.
 
 \item[(b)] We will show that $f$ is WGSP. Suppose not. There exists $S\subseteq N$, $R\in \mathcal{S}_0^n$ and $R'_S\in \mathcal{S}_0^{|S|}$ such that $f(R_S',R_{S})\;P_i\;f(R)$ for all $i\in S$. That means, $f(R_S',R_{-S})\neq f(R)$ and  $R_i\ne R^0$ for any $i\in S$. Let $f(R)=x\neq y=f(R_S',R_{-S})$.  We will consider following three cases.
 
 Case $1$: $x \in [\min \limits_{i\in S}\{\tau(R_i)\},\max \limits_{i\in S}\{\tau(R_i)\}]$. Suppose $x<y$. Let for agent $k\in S$, $\tau(R_k)=\min \limits_{i\in S}\{\tau(R_i)\}$. Then, $f(R)\;P_k\;f(R_S',R_{S})$ - contradicting our assumption. Suppose $y<x$. Let for agent $l\in S$, $\tau(R_l)=\max \limits_{i\in S}\{\tau(R_i)\}$. Then, $f(R)\;P_l\;f(R_S',R_{S})$ - contradicting our assumption.
 
 Case $2$: $x < \min \limits_{i\in S}\{\tau(R_i)\}$. For all $i\in S$, let $R_i^*$ be such that $R_i^*\neq R^0$, $\tau(R_i^*)=\tau(R_i)$ and for all $w< \tau(R_i)$ and $z> \tau(R_i)$, $z\; P_i^* \;w$. Since $f$ is tops-only by Proposition \ref{prop}(a), $f(R^*_S,R_{-S})=f(R)=x$. W.l.o.g., we assume that $S=\{1,\ldots,s\}$. Now applying strategy-proofness, we have that
 
 \begin{align*}
x & =f(R)\\
&= f(R^*_S,R_{-S})\\ 
& \geq f(R'_1,R^*_2,\ldots, R^*_s,R_{-S})\\ 
& \vdots\\
& \geq f(R'_1,R'_2,\ldots, R'_s,R_{-S})\\ 
&= f(R_S',R_{S})\\
&= y
\end{align*}

Since \(x \neq y\), we have $f(R_S', R_{S}) < f(R)$.
If \(x = 0\), this contradicts the fact that $f(R_S', R_{S}) \in [0,1]$. If \(x \neq 0\), then for any \(i \in S\),
$f(R) \; P_i \; f(R_S', R_{S})$, which contradicts our assumption. 
 
 Case $3$: $x > \max \limits_{i\in S}\{\tau(R_i)\}$. For all $i\in S$, let $R_i^*$ be such that $R_i^*\neq R^0$, $\tau(R_i^*)=\tau(R_i)$ and for all $w< \tau(R_i)$ and $z> \tau(R_i)$, $w\; P_i^* \;z$. Since $f$ is tops-only by Proposition \ref{prop}(a), $f(R^*_S,R_{-S})=f(R)=x$. W.l.o.g., we assume that $S=\{1,\ldots,s\}$. Now applying strategy-proofness, we have that
 
 \begin{align*}
x & =f(R)\\
&= f(R^*_S,R_{-S})\\ 
& \leq f(R'_1,R^*_2,\ldots, R^*_s,R_{-S})\\ 
& \vdots\\
& \leq f(R'_1,R'_2,\ldots, R'_s,R_{-S})\\ 
&= f(R_S',R_{S})\\
&= y
\end{align*}

Since \(x \neq y\), we have $f(R)< f(R_S', R_{S})$.
If \(x = 1\), this contradicts the fact that $f(R_S', R_{S}) \in [0,1]$. If \(x \neq 1\), then for any \(i \in S\),
$f(R) \; P_i \; f(R_S', R_{S})$, which contradicts our assumption. 
\end{enumerate}
This completes the proof.
\end{proof}

\subsection{Proof of Theorem \ref{thm}}

\begin{proof} (\textbf{If part})
Let \(f: \mathcal{S}_0^n \longrightarrow A\) be any pairwise strategy-proof and onto SCF. The following claim shows that \(f\) is efficient.

\begin{claim}\rm \label{Claim 1}
\(f\) is efficient.
\end{claim}

\begin{proof} Consider any profile $R \in \mathcal{S}_0^n$. If $R = R^{0N}$, then the set of efficient alternatives is $[0,1]$. Now let $R \neq R^{0N}$. We show that $f(R) \in [\underline{\tau}(R), \overline{\tau}(R)]$. To complete the proof, we consider the following three cases.

\textit{Case 1:} Let $R \neq R^{0N}$ be such that $R_i \neq R^0$ for all $i \in N$, and $\tau(R_j)=\tau(R_k)$ for all $j,k \in N$. W.o.l.g., assume that $\tau(R_1)=\tau(R_2)=\cdots=\tau(R_n)=a$.
We show that $f(R)=a$. Since $f$ is onto, there exists a profile $R' \in \mathcal{S}_0^n$ such that $f(R')=a$. Note that $f(R_1,R'_{-1})=a$. Otherwise, if $f(R_1,R'_{-1})\neq a$, then $f(R')\;P_1\;f(R_1,R'_{-1})$, contradicting strategy-proofness. 
Iterating this replacement argument over agents $2,3,\ldots,n$ and invoking strategy-proofness at each step yields $f(R)=a$.

\textit{Case 2:} Let $R \neq R^{0N}$ be such that $R_i \neq R^0$ for all $i \in N$. We show that $f(R) \in [\underline{\tau}(R), \overline{\tau}(R)]$. If $\underline{\tau}(R) = \overline{\tau}(R)$, the claim follows from Case~1. Hence, suppose that $\underline{\tau}(R) \neq \overline{\tau}(R)$. We assume for contradiction that $f(R)=x \notin [\underline{\tau}(R), \overline{\tau}(R)]$. We consider the following two subcases.

\textit{Subcase 2.1:} $x < \underline{\tau}(R)$. Consider a profile $R^* \in \mathcal{S}_0^n$ such that, for all $i \in N$, $R_i^*\neq R^0$, $\tau(R_i^*) = \underline{\tau}(R)$, and for all alternatives $w,z \in A$ satisfying $w < \underline{\tau}(R) < z$,  $w \; P_i^* \; z$. Starting from $R$, replace agents' preferences one by one by $R_i^*$. By repeated applications of strategy-proofness, the outcome is unchanged along this sequence; hence $f(R^*)=f(R)=x$. However, by Case~1, we have $f(R^*) = \underline{\tau}(R)$, which contradicts $x < \underline{\tau}(R)$.

\textit{Subcase 2.2:} $x > \overline{\tau}(R)$. Consider a profile $R^{**} \in \mathcal{S}_0^n$ such that, for all $i \in N$, $R_i^{**}\neq R^0$, $\tau(R_i^{**}) = \overline{\tau}(R)$, and for all alternatives $w,z \in A$ satisfying $w < \overline{\tau}(R) < z$, $z \; P_i^{**} \; w$. Starting from $R$, replace agents' preferences one by one by $R_i^{**}$. By repeated applications of strategy-proofness, the outcome is unchanged along this sequence; hence $f(R^{**})=f(R)=x$. However, by Case~1, we have $f(R^{**}) = \overline{\tau}(R)$, which contradicts $x > \overline{\tau}(R)$.

\textit{Case 3:} Let $R \neq R^{0N}$ such that for some $S\subset N$, and for for all $i\in S$, $R_i=R^0$. W.o.l.g., let $S=\{1,\ldots,s\}$ and $N\setminus S=\{s+1,\ldots,n\}$. We show that $f(R) \in [\underline{\tau}(R), \overline{\tau}(R)]$. W.o.l.g., we assume that $\tau(R_{s+1})=\underline{\tau}(R)$ and $\tau(R_{n})=\overline{\tau}(R)$. Consider a profile $R''\in \mathcal{S}_0^n$ such that for all $i\in S$, $R_i''\neq R^0$ and $\tau(R_i'')=\underline{\tau}(R)$ and for all $i\notin S$, $R_i''=R_i$.  From case~2, we have that $f(R'') \in [\underline{\tau}(R), \overline{\tau}(R)]$.  Now at $R''$, we replace agent~1's preference by $R_1$ and claim that  $f(R_1,R_2''\ldots,R_n'')\in [\underline{\tau}(R),\overline{\tau}(R)]$. Suppose not. Then, either $f(R_1,R_2''\ldots,R_n'')<\underline{\tau}(R)$ or $\overline{\tau}(R)<f(R_1,R_2''\ldots,R_n'')$. Suppose $f(R_1,R_2''\ldots,R_n'')<\underline{\tau}(R)$. Consider the coalition $\{1,n\}$, the profile $(R_1,R_2''\ldots,R_n'')$ and the profile for the coalition $\{1,n\}$, $(R''_1,R''_n)$. Note that $$f(R''_1,R''_n,R''_{-\{1,n\}})\;I_1\;f(R_1,R_2''\ldots,R_n'')$$ and  $f(R''_1,R''_n,R''_{-\{1,n\}})\;P''_n\;f(R_1,R_2''\ldots,R_n'')$. This contradicts Pairwise strategy-proofness. Now, suppose $\overline{\tau}(R)<f(R_1,R''_2\ldots,R''_n)$. Now consider the coalition $\{1,s+1\}$, the profile $(R_1,R_2''\ldots,R_n'')$ and the profile for the coalition $\{1,s+1\}$, $(R''_1,R''_{s+1})$. Note that $$f(R''_1,R''_{s+1},R''_{-\{1,s+1\}})\;I_1\;f(R_1,R_2''\ldots,R_n'')$$ and  $f(R''_1,R''_{s+1},R''_{-\{1,s+1\}})\;P''_{s+1}\;f(R_1,R_2''\ldots,R_n'')$. This contradicts Pairwise strategy-proofness. Therefore, $f(R_1,R_2''\ldots,R_n'')\in [\underline{\tau}(R),\overline{\tau}(R)]$. Iterating this replacement argument over agents $2,3,\ldots,s$ and invoking pairwise strategy-proofness at each step yields $$f(R_1,\ldots,R_s,R''_{s+1},\ldots, R''_n)\in [\underline{\tau}(R), \overline{\tau}(R)].$$ Note that $R=(R_1,\ldots,R_s,R''_{s+1},\ldots, R''_n)$. Hence, $f(R)\in [\underline{\tau}(R), \overline{\tau}(R)]$.
\end{proof}

The next claim shows that \(f\) satisfies tops-onlyness.

\begin{claim}\rm \label{Claim 2}
$f$ is tops-only.
\end{claim}

\begin{proof}
From Claim~\ref{Claim 1}, it follows that $f$ is efficient. Moreover, since $f$ is pairwise strategy-proof, it is strategy-proof. Hence, Proposition~\ref{prop} implies that $f$ is tops-only.
\end{proof}

Before proceeding, we note the following claim, which will be used in what follows.

\begin{claim}\rm \label{Claim 3}
Let $R \in \mathcal{S}_0^n$ be such that, for some agents $i,j \in N$, we have $R_i, R_j \neq R^0$, $\tau(R_i) = 0$, and $\tau(R_j) = 1$. Then, for any agent $k \in N \setminus \{i,j\}$ and any $R'_k \in \mathcal{S}_0$, we have $f(R'_k, R_{-k}) = f(R)$.
\end{claim}

\begin{proof} Fix a profile $R \in \mathcal{S}_0^n$ and agents $i,j \in N$ such that $R_i, R_j \neq R^0$, $\tau(R_i)=0$, and $\tau(R_j)=1$. Let $k \in N \setminus \{i,j\}$ be arbitrary, and let $R'_k \in \mathcal{S}_0$ be any preference. We show that $f(R'_k, R_{-k}) = f(R)$. We consider following cases.

Case $1$: $R_k\neq R^0$ and $R'_k=R^0$. Suppose $f(R)<f(R'_k, R_{-k})$. Consider coalition $\{i,k\}$, the profile $(R_i,R'_k,R_{-\{i,k\}})$ and the profile for the coalition $\{i,k\}$, $(R_i,R_k)$. Note that $$f(R_i,R_k,R_{-\{i,k\}})\;P_i\;f(R_i,R'_k,R_{-\{i,k\}})$$ and $f(R_i,R_k,R_{-\{i,k\}})\;I'_k\;f(R_i,R'_k,R_{-\{i,k\}})$. This contradicts Pairwise strategy-proofness. Suppose $f(R'_k, R_{-k})<f(R)$. Now consider coalition $\{j,k\}$, the profile $(R_j,R'_k,R_{-\{j,k\}})$ and the profile for the coalition $\{j,k\}$, $(R_j,R_k)$. Note that $$f(R_j,R_k,R_{-\{j,k\}})\;P_j\;f(R_j,R'_k,R_{-\{j,k\}})$$ and $f(R_j,R_k,R_{-\{j,k\}})\;I'_k\;f(R_j,R'_k,R_{-\{j,k\}})$. This contradicts Pairwise strategy-proofness. Therefore, $f(R'_k, R_{-k}) = f(R)$.

Case $2$: $R_k= R^0$ and $R'_k\neq R^0$. Suppose $f(R'_k, R_{-k})<f(R)$. Consider coalition $\{i,k\}$, the profile $(R_i,R_k,R_{-\{i,k\}})$ and the profile for the coalition $\{i,k\}$, $(R_i,R'_k)$. Note that $$f(R_i,R'_k,R_{-\{i,k\}})\;P_i\;f(R_i,R_k,R_{-\{i,k\}})$$ and $f(R_i,R'_k,R_{-\{i,k\}})\;I_k\;f(R_i,R_k,R_{-\{i,k\}})$. This contradicts Pairwise strategy-proofness. Suppose $f(R)<f(R'_k, R_{-k})$. Now consider coalition $\{j,k\}$, the profile $(R_j,R_k,R_{-\{j,k\}})$ and the profile for the coalition $\{j,k\}$, $(R_j,R'_k)$. Note that $$f(R_j,R'_k,R_{-\{j,k\}})\;P_j\;f(R_j,R_k,R_{-\{j,k\}})$$ and $f(R_j,R'_k,R_{-\{j,k\}})\;I_k\;f(R_j,R_k,R_{-\{j,k\}})$. This contradicts Pairwise strategy-proofness. Therefore, $f(R'_k, R_{-k}) = f(R)$.

Case $3$: $R_k\neq R^0$ and $R'_k\neq R^0$. First, at $R$, we replace agent $k$'s preference by $R_k''$ where $R_k''=R^0$. By the argument used in Case~1, it follows that $f(R''_k, R_{-k}) = f(R)$. Next, at at the profile $(R''_k, R_{-k})$, replace agent $k$'s preference by $R_k'$. By the argument used in Case~2, we obtain $f(R'_k, R_{-k}) = f(R''_k, R_{-k})$. Therefore, $f(R'_k, R_{-k}) = f(R''_k, R_{-k})=f(R)$.
\end{proof}

We now proceed to complete the proof. Let $\bar{R} \in \mathcal{S}_0^n$ be such that $\bar{R}_i \neq R^0$ for all $i \in N$, $\tau(\bar{R}_i) = 0$ for all $i \in \{1,2,\ldots,n-1\}$, and $\tau(\bar{R}_n) = 1$. Let $f(\bar{R}) = x$ and $f(R^{0N}) = y$. We will show that $f$ is a target rule with default, where $x$ is the target level and $y$ is the default alternative. The following claims complete the proof.

\begin{claim}\rm \label{Claim 4}
Let $i \in N$ be any agent, and let $R = (R_i, R_{-i}) \in \mathcal{S}_0^n$ be a preference profile such that $R_k \neq R^0$ for all $k \in N$, $\tau(R_i) = 1$, and $\tau(R_j) = 0$ for all $j \in N \setminus \{i\}$. Then $f(R) = x$.
\end{claim}

\begin{proof} Fix any agent $i \in N$ and consider a profile 
$R = (R_i, R_{-i}) \in \mathcal{S}_0^n$ such that 
$R_k \neq R^0$ for all $k \in N$, $\tau(R_i) = 1$, and 
$\tau(R_j) = 0$ for all $j \in N \setminus \{i\}$. We show that $f(R)=x$. If $i = n$, then, $R=\bar{R}$. Hence, $f(R)=x$. So, we assume that $i \neq n$. 

At the profile $\bar{R}$, we replace agent $i$'s preference by $R_i$. From Claim \ref{Claim 3}, we have that $f(R_i, \bar{R}_{-i})=f(\bar{R})=x$. Now, at the profile $(R_i, \bar{R}_{-i})$, we replace agent $n$'s preference by $R_n$. From Claim \ref{Claim 3}, we have that $f(R_i,R_n, \bar{R}_{-\{i,n\}})=f(R_i, \bar{R}_{-i})$. Therefore, $f(R_i,R_n, \bar{R}_{-\{i,n\}})=x$. Note that $R=(R_i,R_n, \bar{R}_{-\{i,n\}})$. Hence, $f(R)=f(R_i,R_n, \bar{R}_{-\{i,n\}})=x$.
\end{proof}

\begin{claim}\rm \label{Claim 5}
Let $R \in \mathcal{S}_0^n$ with $R \neq R^{0N}$ and
$\underline{\tau}(R) \leq x \leq \overline{\tau}(R)$. Then $f(R) = x$.
\end{claim}
\begin{proof} Let $R \in \mathcal{S}_0^n$ with $R \neq R^{0N}$ and
$\underline{\tau}(R) \leq x \leq \overline{\tau}(R)$. We show that $f(R)=x$. If $\underline{\tau}(R)=\overline{\tau}(R)$, then by efficiency of $f$ (Claim~\ref{Claim 1}), we have $f(R)=x$. Hence, assume that $\underline{\tau}(R)\neq \overline{\tau}(R)$.
Let $i,j \in N$ be such that $R_i, R_j \neq R^0$, $\tau(R_i)=\underline{\tau}(R)$, and $\tau(R_j)=\overline{\tau}(R)$.

We begin with a profile $R' \in \mathcal{S}_0^n$ such that $R'_k \neq R^0$ for all $k \in N$, $\tau(R'_j) = 1$, and $\tau(R'_l) = 0$ for all $l \in N \setminus \{j\}$. By Claim~\ref{Claim 4}, we have $f(R') = x$.

Starting from $R'$, we replace the preferences of agents in $N \setminus \{i,j\}$ sequentially, one agent at a time, by their corresponding preferences in $R$. By Claim~\ref{Claim 3}, the outcome remains unchanged at each replacement step. After all such replacements are completed, we obtain $f(R'_i, R'_j, R_{-\{i,j\}}) = x$.

Now consider the profile $(R'_i, R'_j, R_{-\{i,j\}})$ and replace agent $i$'s preference by $R''_i$, where $R''_i \neq R^0$, $\tau(R''_i)=\underline{\tau}(R)$, and for all $w<\underline{\tau}(R)$ and $z>\underline{\tau}(R)$, we have $w \; P''_i \; z$. If $f(R''_i, R'_j, R_{-\{i,j\}}) < x$, then
\[
f(R''_i, R'_j, R_{-\{i,j\}})\; P'_i \; f(R'_i, R'_j, R_{-\{i,j\}}),
\]
contradicting strategy-proofness. If $x < f(R''_i, R'_j, R_{-\{i,j\}})$, then
\[
f(R'_i, R'_j, R_{-\{i,j\}})\; P''_i \; f(R''_i, R'_j, R_{-\{i,j\}}),
\]
again contradicting strategy-proofness. Therefore, $f(R''_i, R'_j, R_{-\{i,j\}}) = x$.

Next, at $(R''_i, R'_j, R_{-\{i,j\}})$, replace agent $j$'s preference by $R''_j$, where $R''_j \neq R^0$, $\tau(R''_j)=\overline{\tau}(R)$, and for all $w<\overline{\tau}(R)$ and $z>\overline{\tau}(R)$, we have $z \; P''_j \; w$. If $f(R''_i, R''_j, R_{-\{i,j\}}) < x$, then
\[
f(R''_i, R'_j, R_{-\{i,j\}})\; P''_j \; f(R''_i, R''_j, R_{-\{i,j\}}),
\]
contradicting strategy-proofness. If $x < f(R''_i, R''_j, R_{-\{i,j\}})$, then
\[
f(R''_i, R''_j, R_{-\{i,j\}})\; P'_j \; f(R''_i, R'_j, R_{-\{i,j\}}),
\]
again contradicting strategy-proofness. Hence, $f(R''_i, R''_j, R_{-\{i,j\}}) = x$.

Finally, by tops-onlyness of $f$ (Claim~\ref{Claim 2}), we obtain
$f(R''_i, R''_j, R_{-\{i,j\}}) = f(R)$, and therefore $f(R) = x$.
\end{proof}

\begin{claim}\rm \label{Claim 6}
Let $R \in \mathcal{S}_0^n$ with $R \neq R^{0N}$ and
$x < \underline{\tau}(R)$. Then $f(R) = \underline{\tau}(R)$.
\end{claim}
\begin{proof} Let $R \in \mathcal{S}_0^n$ with $R \neq R^{0N}$ and
$x <\underline{\tau}(R)$. We show that $f(R)=\underline{\tau}(R)$. If $\underline{\tau}(R)=\overline{\tau}(R)$, then by efficiency of $f$ (Claim~\ref{Claim 1}), we have $f(R)=\underline{\tau}(R)$. Hence, assume that $\underline{\tau}(R)\neq \overline{\tau}(R)$.
Let $i,j \in N$ be such that $R_i, R_j \neq R^0$, $\tau(R_i)=\underline{\tau}(R)$, and $\tau(R_j)=\overline{\tau}(R)$.

We begin with a profile $R' \in \mathcal{S}_0^n$ such that $R'_k \neq R^0$ for all $k \in N$, $\tau(R'_j) = 1$, and $\tau(R'_l) = 0$ for all $l \in N \setminus \{j\}$. By Claim~\ref{Claim 4}, we have $f(R') = x$.

Starting from $R'$, we replace the preferences of agents in $N \setminus \{i,j\}$ sequentially, one agent at a time, by their corresponding preferences in $R$. By Claim~\ref{Claim 3}, the outcome remains unchanged at each replacement step. After all such replacements are completed, we obtain $f(R'_i, R'_j, R_{-\{i,j\}}) = x$.

Now consider the profile $(R'_i, R'_j, R_{-\{i,j\}})$ and replace agent $j$'s preference by $R''_j$, where $R''_j \neq R^0$, $\tau(R''_j)=\overline{\tau}(R)$, and for all $w<\overline{\tau}(R)$ and $z>\overline{\tau}(R)$, we have $z \; P''_j \; w$. If $f(R'_i, R''_j, R_{-\{i,j\}}) < x$, then
\[
f(R'_i, R'_j, R_{-\{i,j\}})\; P''_j \; f(R'_i, R''_j, R_{-\{i,j\}}),
\]
contradicting strategy-proofness. If $x < f(R'_i, R''_j, R_{-\{i,j\}})$, then
\[
f(R'_i, R''_j, R_{-\{i,j\}})\; P'_j \; f(R'_i, R'_j, R_{-\{i,j\}}),
\]
again contradicting strategy-proofness. Therefore, $f(R'_i, R''_j, R_{-\{i,j\}}) = x$.

Next, at $(R'_i, R''_j, R_{-\{i,j\}})$, replace agent $i$'s preference by $R''_i$, where $R''_i \neq R^0$, $\tau(R''_i)=\underline{\tau}(R)$, and for all $w<\underline{\tau}(R)$ and $z>\underline{\tau}(R)$, we have $w \; P''_i \; z$. Note that
\[
\underline{\tau}((R''_i, R''_j, R_{-\{i,j\}}))=\underline{\tau}(R).
\]
By efficiency of $f$ (Claim~\ref{Claim 1}), it cannot be that
$f(R''_i, R''_j, R_{-\{i,j\}}) < \underline{\tau}(R)$. If
$\underline{\tau}(R) < f(R''_i, R''_j, R_{-\{i,j\}})$, then
\[
f(R'_i, R''_j, R_{-\{i,j\}})\; P''_i \;
f(R''_i, R''_j, R_{-\{i,j\}}),
\]
contradicting strategy-proofness. Hence, $f(R''_i, R''_j, R_{-\{i,j\}}) = \underline{\tau}(R)$.

Finally, by tops-onlyness of $f$ (Claim~\ref{Claim 2}), we obtain
$f(R''_i, R''_j, R_{-\{i,j\}}) = f(R)$, and therefore $f(R) = \underline{\tau}(R)$.
\end{proof}

\begin{claim}\rm \label{Claim 7}
Let $R \in \mathcal{S}_0^n$ with $R \neq R^{0N}$ and
$\overline{\tau}(R) < x$. Then $f(R) = \overline{\tau}(R)$.
\end{claim}
\begin{proof}Let $R \in \mathcal{S}_0^n$ with $R \neq R^{0N}$ and
$\overline{\tau}(R)< x $. We show that $f(R)=\overline{\tau}(R)$. If $\underline{\tau}(R)=\overline{\tau}(R)$, then by efficiency of $f$ (Claim~\ref{Claim 1}), we have $f(R)=\overline{\tau}(R)$. Hence, assume that $\underline{\tau}(R)\neq \overline{\tau}(R)$.
Let $i,j \in N$ be such that $R_i, R_j \neq R^0$, $\tau(R_i)=\underline{\tau}(R)$, and $\tau(R_j)=\overline{\tau}(R)$.

We begin with a profile $R' \in \mathcal{S}_0^n$ such that $R'_k \neq R^0$ for all $k \in N$, $\tau(R'_j) = 1$, and $\tau(R'_l) = 0$ for all $l \in N \setminus \{j\}$. By Claim~\ref{Claim 4}, we have $f(R') = x$.

Starting from $R'$, we replace the preferences of agents in $N \setminus \{i,j\}$ sequentially, one agent at a time, by their corresponding preferences in $R$. By Claim~\ref{Claim 3}, the outcome remains unchanged at each replacement step. After all such replacements are completed, we obtain $f(R'_i, R'_j, R_{-\{i,j\}}) = x$.

Now consider the profile $(R'_i, R'_j, R_{-\{i,j\}})$ and replace agent $i$'s preference by $R''_i$, where $R''_i \neq R^0$, $\tau(R''_i)=\underline{\tau}(R)$, and for all $w<\underline{\tau}(R)$ and $z>\underline{\tau}(R)$, we have $w \; P''_i \; z$. If $f(R''_i, R'_j, R_{-\{i,j\}}) < x$, then
\[
f(R''_i, R'_j, R_{-\{i,j\}})\; P'_i \; f(R'_i, R'_j, R_{-\{i,j\}}),
\]
contradicting strategy-proofness. If $x < f(R''_i, R'_j, R_{-\{i,j\}})$, then
\[
f(R'_i, R'_j, R_{-\{i,j\}})\; P''_i \; f(R''_i, R'_j, R_{-\{i,j\}}),
\]
again contradicting strategy-proofness. Therefore, $f(R''_i, R'_j, R_{-\{i,j\}}) = x$.

Next, at $(R''_i, R'_j, R_{-\{i,j\}})$, replace agent $j$'s preference by $R''_j$, where $R''_j \neq R^0$, $\tau(R''_j)=\overline{\tau}(R)$, and for all $w<\overline{\tau}(R)$ and $z>\overline{\tau}(R)$, we have $z \; P''_j \; w$.  Note that
\[
\overline{\tau}((R''_i, R''_j, R_{-\{i,j\}}))=\overline{\tau}(R).
\]
By efficiency of $f$ (Claim~\ref{Claim 1}), it cannot be that
$\overline{\tau}(R)< f(R''_i, R''_j, R_{-\{i,j\}})$. If $f(R''_i, R''_j, R_{-\{i,j\}}) < \overline{\tau}(R)$, then
\[
f(R''_i, R'_j, R_{-\{i,j\}})\; P''_j \; f(R''_i, R''_j, R_{-\{i,j\}}),
\]
contradicting strategy-proofness. Hence, $f(R''_i, R''_j, R_{-\{i,j\}}) = \overline{\tau}(R)$.

Finally, by tops-onlyness of $f$ (Claim~\ref{Claim 2}), we obtain
$f(R''_i, R''_j, R_{-\{i,j\}}) = f(R)$, and therefore $f(R) = \overline{\tau}(R)$.
\end{proof}

(\textbf{Only-if part}) Let $f : \mathcal{S}_0^n \to A$ be a target rule with a default, with target level $x$ and default alternative $y$. We show that $f$ is pairwise strategy-proof. In fact, the following lemma establishes the stronger property that $f$ is group strategy-proof, from which pairwise strategy-proofness follows immediately.

\begin{lemma}\label{L2}
$f$ is group strategy-proof.
\end{lemma}

\begin{proof} Assume, for contradiction, that the statement is false. Then there exist a profile $R \in \mathcal{S}_0^n$, a coalition $S \subseteq N$, and $R'_S$ such that $f(R'_S, R_{-S}) \; R_i \; f(R)$ for all $i \in S$, and $f(R'_S, R_{-S}) \; P_j \; f(R)$ for some $j \in S$. Let $f(R) = w$ and $f(R'_S, R_{-S}) = z$. By our assumption that $z \; P_j \; w$, we have $R_j\neq R^0$ and $R \neq R^{0N}$. Since $x$ is the target level, one of the following cases must hold: (1) $\underline{\tau}(R)\leq x = w \leq \overline{\tau}(R)$, or (2) $x < \underline{\tau}(R) = w \leq \overline{\tau}(R)$, or (3) $\underline{\tau}(R)\leq w = \overline{\tau}(R) < x$. We consider these three cases to complete the proof.

\textit{Case 1:} $\underline{\tau}(R)\leq x = w \leq \overline{\tau}(R)$. First we consider that $(R_S',R_{-S})=R^{0N}$. Suppose $w<z$. Note that there exists an agent $i\in S$ such that $R_i\neq R^0$ and $\tau(R_i)=\underline{\tau}(R)$. For that agent $f(R)\;P_i\;f(R_S',R_{-S})$ because $\tau(R_i)\leq w<z$. This contradicts our assumption that $f(R'_S,R_{-S}) \; R_i \; f(R)$.   Suppose $z<w$. Note that there exists an agent $k\in S$ such that $R_k\neq R^0$ and $\tau(R_k)=\overline{\tau}(R)$. For that agent $f(R)\;P_k\;f(R_S',R_{-S})$ because $z<w\leq \tau(R_k)$.  Again, this contradicts our assumption that $f(R'_S,R_{-S}) \; R_k \; f(R)$.

Now we consider that $(R_S',R_{-S})\ne R^{0N}$. Let $w<z$. Since $x$ is the target level, we have that $f(R_S',R_{-S})=\underline{\tau}(R_S',R_{-S})=z$. Then there exists an agent $i\in S$ such that $R_i\neq R^0$ and $\tau(R_i)=\underline{\tau}(R)$. For that agent $f(R)\;P_i\;f(R_S',R_{-S})$ because $\tau(R_i)\leq w<z$, a contradiction. Let $z<w$. Since $x$ is the target level, we have that $f(R_S',R_{-S})=\overline{\tau}(R_S',R_{-S})=z$. Then there exists an agent $k\in S$ such that $R_k\neq R^0$ and $\tau(R_k)=\overline{\tau}(R)$. For that agent $f(R)\;P_k\;f(R_S',R_{-S})$ because $z<w\leq \tau(R_k)$, a contradiction. 

\textit{Case 2:} $x < \underline{\tau}(R) = w \leq \overline{\tau}(R)$. First we consider that $(R_S',R_{-S})=R^{0N}$. Suppose $w<z$. Note that there exists an agent $i\in S$ such that $R_i\neq R^0$ and $\tau(R_i)=\underline{\tau}(R)$. For that agent $f(R)\;P_i\;f(R_S',R_{-S})$ because $\tau(R_i)\leq w<z$, a contradiction.  Suppose $z<w$. Then, for agent $j\in S$, $f(R)\;P_j\;f(R_S',R_{-S})$ because $z<w\leq \tau(R_j)$, a contradiction. 

Now we consider that $(R_S',R_{-S})\ne R^{0N}$. Let $w<z$. Since $x$ is the target level, we have that $f(R_S',R_{-S})=\underline{\tau}(R_S',R_{-S})=z$. Then there exists an agent $i\in S$ such that $R_i\neq R^0$ and $\tau(R_i)=\underline{\tau}(R)$. For that agent $f(R)\;P_i\;f(R_S',R_{-S})$ because $\tau(R_i)\leq w<z$, a contradiction. Let $z<w$. Then, for agent $j\in S$, $f(R)\;P_j\;f(R_S',R_{-S})$ because $z<w\leq \tau(R_j)$, a contradiction.

\textit{Case 3:} $\underline{\tau}(R)\leq w = \overline{\tau}(R) < x$. First we consider that $(R_S',R_{-S})=R^{0N}$. Suppose $w<z$. Then, for agent $j\in S$, $f(R)\;P_j\;f(R_S',R_{-S})$ because $\tau(R_j)\leq w<z$, a contradiction. Suppose $z<w$. Then there exists an agent $k\in S$ such that $R_k\neq R^0$ and $\tau(R_k)=\overline{\tau}(R)$. For that agent $f(R)\;P_k\;f(R_S',R_{-S})$ because $z<w\leq \tau(R_k)$, a contradiction.  

Now we consider that $(R_S',R_{-S})\ne R^{0N}$. Let $w<z$. In this case, for agent $j\in S$, $f(R)\;P_j\;f(R_S',R_{-S})$ because $\tau(R_j)\leq w<z$, a contradiction. Let $z<w$. Since $x$ is the target level, we have that $f(R_S',R_{-S})=\overline{\tau}(R_S',R_{-S})=z$. Then there exists an agent $k\in S$ such that $R_k\neq R^0$ and $\tau(R_k)=\overline{\tau}(R)$. For that agent $f(R)\;P_k\;f(R_S',R_{-S})$ because $z<w\leq \tau(R_k)$, a contradiction.

Therefore, $f$ is group strategy-proof. 
\end{proof}

This completes the proof of the theorem.
\end{proof}

\end{document}